\documentclass[11pt]{article} 
\usepackage{times}
\usepackage{amsmath}
\usepackage{amssymb}
\usepackage{algorithm2e}
\usepackage{graphicx,epsfig}

\newcommand {\E}{\mathbf E}
\newcommand {\T}{\mathcal T}

\newcommand{\p}{\mathbf{p}}
\newcommand{\eat}[1]{}
\newtheorem{theorem}{Theorem}[section] 
\newtheorem{lemma}[theorem]{Lemma}

\newtheorem{definition}{Definition}

\newtheorem{ex}{Example}[section]

\newenvironment{proof}{{\bf Proof:}}{ \hfill $\Box$}
\renewcommand{\paragraph}[1]{\medskip \noindent {\bf{#1}}}

\advance \topmargin by -\headheight
\advance \topmargin by -\headsep
\evensidemargin \oddsidemargin
\renewcommand{\O}{\mathcal{O}}

\newcommand{\Prob}{\mathbb{P}}

\begin{document}

\title{Graph Sparsification in the Semi-streaming Model}
\author{Kook Jin Ahn\and
Sudipto Guha\thanks{
Department of Computer and Information Science, University of Pennsylvania,
Philadelphia PA 19104-6389.
Email:{\tt \{kookjin,sudipto\}@cis.upenn.edu}.
Research supported in part by an Alfred
 P. Sloan Research Fellowship, and NSF Awards
CCF-0644119 and IIS-0713267.} }
\maketitle

\thispagestyle{empty}

\begin{abstract}
  Analyzing massive data sets has been one of the key motivations for
  studying streaming algorithms. In recent years, there has been
  significant progress in analysing distributions in a streaming setting,
  but the progress on graph problems has been limited. A main reason for
  this has been the existence of linear space lower bounds for even simple
  problems such as determining the connectedness of a graph. However, in
  many new scenarios that arise from social and other interaction
  networks, the number of vertices is significantly less than
  the number of edges. This has led to the formulation of the
  semi-streaming model where we assume that the space is (near) linear in the
  number of vertices (but not necessarily the edges), and the edges appear
  in an arbitrary (and possibly adversarial) order.

  However there has been limited progress in analysing graph algorithms in
  this model. In this paper we focus on graph sparsification, which is one of the
  major building blocks in a variety of graph algorithms. Further,
  there has been a long history of (non-streaming) sampling algorithms that provide
  sparse graph approximations and it a natural question to ask: since the
  end result of the sparse approximation is a small (linear) space
  structure, can we achieve that using a small space, and in addition
  using a single pass over the data? The question is interesting from the
  standpoint of both theory and practice and we answer the question in the
  affirmative, by providing a one pass $\tilde{O}(n/\epsilon^{2})$ space
  algorithm that produces a sparsification that approximates each cut to a
  $(1+\epsilon)$ factor. 
  We also show that $\Omega(n \log \frac1\epsilon)$ space is
  necessary for a one pass streaming algorithm to approximate the min-cut,
  improving upon the $\Omega(n)$ lower bound that arises from lower bounds for testing
  connectivity.
\end{abstract}

\newcommand{\Th}{\Theta}
\renewcommand{\S}{\mathcal{S}}
\renewcommand{\H}{\mathbf{T}}
\renewcommand{\T}{\mathbf{T}}
\renewcommand{\p}{\mathbf{p}}
\newcommand{\q}{\mathbf{q}}
\newcommand{\Eta}{\mathcal{E}}
\newcommand{\Z}{\mathcal{Z}}
\newcommand{\Y}{\mathcal{Y}}
\renewcommand{\L}{\mathcal{L}}

\section{Introduction}
The feasibility of processing graphs in the data stream model was one of
the early questions investigated in the streaming model
\cite{henzinger1998}.  However the results were not encouraging, even to
decide simple properties such as the connectivity of a graph, when the
edges are streaming in an arbitrary order required $\Omega(n)$ space. In
comparison to the other results in the streaming model, \cite{AMS96,MP80}
which required polylogarithmic space, graph alogithms appeared to
difficult in the streaming context and did not receive much attention
subsequently.

However in recent years, with the remergence of social and other
interaction networks, questions of processing massive graphs have once
again become prominent. Technologically, since the publication of
\cite{henzinger1998}, it had become feasible to store larger quantities of
data in memory and the semi-streaming model was proposed in
\cite{mcgregor2005,muthu}. In this model we assume that the space is (near) linear
in the number of vertices (but not necessarily the edges). Since its
formulation, the model has become more appealing from the contexts of
theory as well as practice.  From a theoretical viewpoint, the model still
offers a rich potential trade-off between space and
accuracy of algorithm, albeit at a different threshold than
polylogarithmic space. From a practical standpoint, in a variety of contexts involving large
graphs, such as image segmentation using graph cuts, the ability of the
algorithm to retain the most relevant information in main memory has been
deemed critical. In essence, an algorithm that runs out of main memory
space would become unattractive and infeasible. In such a setting, it may
be feasible to represent the vertex set in the memory whereas the edge
set may be significantly larger.

In the semi-streaming model, the first results were provided by
\cite{mcgregor2005} on the construction of graph spanners.  Subsequently,
beyond explorations of connectivity \cite{demetrescu2006}, and (multipass)
matching \cite{mcgregor2005b}, there has been little development of
algorithms in this model.  In this paper we focus on the problem of graph
sparsification in a single pass, that is, constructing a small space
representation of the graph such that we can estimate the size of any cut.
Graph sparsification \cite{benczur1996,spielman2008} remains one of the
major building blocks for a variety of graph algorithms, such as flows and
disjoint paths, etc. At the same time, sparsification immediately provides
a way of finding an approximate min-cut in a graph.  The problem of
finding a min-cut in a graph has been one of the more celebrated problems and
there is a vast literature on this problem, including both 
deterministic ~\cite{gomory1961,hao1992} as well as randomized
algorithms~\cite{karger1993,karger1994,karger1996,karger2000} -- see
\cite{chekuri1997} for a comprehensive discussion of various algorithms.
We believe that a result on sparsification will
enable the investigation of a richer class of problems in graphs in the
semi-streaming model.

In this paper we will focus exclusively on the model that the stream is
adversarially ordered and a single pass is allowed. From the standpoint of
techniques, our algorithm is similar in spirit to the algorithm of
Alon-Matias-Szegedy \cite{AMS96},
where we simultaneously sample and estimate from the stream. In fact we
show that in the semi-streaming model we can perform a similar, but
non-trivial, simultaneous sampling and
estimation. This is pertinent because sampling algorithms for
sparsification exist~\cite{benczur1996,spielman2008}, which use $\O(n
\mathrm{polylog}(n))$ edges. However these algorithms sample edges in an iterative
fashion that requires the edges to be present in memory and random access
to them. 

\paragraph{Our Results:}
Our approach is to
recursively maintain a summary of the graph seen so far and use that
summary itself to decide on the action to be taken on seeing a new edge.
To this end, we modify the sparsification algorithm of Benczur and
Karger~\cite{benczur1996} for the semi--streaming model. The final
algorithm uses a single pass over the edges and provides $1 \pm \epsilon$
approximation for cut values with high probability and uses $\O(n(\log n +
\log m)(\log \frac{m}{n})(1+\epsilon)^2/\epsilon^2)$ edges for $n$ node and
$m$ edge graph.

\section{Background and Notation}
Let $G$ denote the input graph and $n$ and $m$ respectively denote the
number of nodes and edges. $VAL(C,G)$ denotes the value of cut $C$ in $G$.
$w_G(e)$ indicates the weight of $e$ in graph $G$.

\begin{definition}
\cite{benczur1996} A graph is \textbf{k-strong connected} if and only if every cut in the graph has value at least $k$. \textbf{k-strong connected component} is a maximal node-induced subgraph which is k-strong connected. The \textbf{strong connectivity} of an edge $e$ is the maximum $k$ such that there exists a $k$-strong connected component that contains $e$.
\end{definition}

In \cite{benczur1996}, they compute the strong connectivity of each edge and use it to decide the sampling probability. Algorithm \ref{alg:sparsify} is their algorithm. We will modify this in section \ref{sec:algdesc}.

\medskip
\begin{algorithm}[H]\label{alg:sparsify}
\noindent {\bf Benczur-Karger}(\cite{benczur1996})~\\
  \SetLine
  \KwData{Graph $G=(V,E)$}
  \KwResult{Sparsified graph $H$}
    compute the strong connectivity of edge $c^{G}_e$ for all $e\in G$\;
    $H\leftarrow(V,\emptyset)$\;
    \ForEach{$e$} {
      $p_e=\min\{\rho/c_e,1\}$\;
      with probability $p_e$, add $e$ to $H$ with weight $1/p_e$\;
    }
  \caption{Sparsification Algorithm}
\end{algorithm}

Here $\rho$ is a parameter that depends on the size of $G$ and the error
bound $\epsilon$. They proved the following two theorems in their paper.

\begin{theorem}\label{thm:benczur_errorbound}~\cite{benczur1996} Given $\epsilon$ and a corresponding $\rho=16(d+2)(\ln n)/\epsilon^2$, every cut in $H$ has value between $(1-\epsilon)$ and $(1+\epsilon)$ times its value in $G$ with probability $1-n^{-d}$.
\end{theorem}

\begin{theorem}\label{thm:benczur_spacebound}~\cite{benczur1996} With high probability $H$ has $\O(n\rho)$ edges.
\end{theorem}

\newcommand{\hijp}{H_{i_j-1}}
\newcommand{\hij}{H_{i_j}}
\newcommand{\hijs}{H_{i_{j+1}-1}}

\newcommand{\peij}{p_{e_{i_j}}}

Throughout this paper, $e_1,e_2,\cdots,e_m$ denotes the input sequence.
$G_i$ is a graph that consists of $e_1$,$e_2$,$\cdots$,$e_i$. $c^{(G)}_e$ is the
strong connectivity of $e$ in $G$ and $w_G(e)$ is weight of an edge $e$ in
$G$. $G_{i,j}=\{e:e\in G_i,2^{j-1}\leq c^{(G_i)}_e < 2^j\}$. Each edge has
weight 1 in $G_{i,j}$. $F_{i,j}=\sum_{k\geq j} 2^{j-k}G_{i,j}$ where
scalar multiplication of a graph and addition of a graph is defined as
scalar multiplication and addition of edge weights. In addition, $H\in(1\pm\epsilon)G$ if and only if $(1-\epsilon)VAL(C,G)\leq VAL(C,H)\leq
(1+\epsilon)VAL(C,G)$. $H_i$ is a sparsification of a graph $G_i$, i.e., a
sparsified graph after considering $e_i$ in the streaming model.

\section{A Semi-Streaming Algorithm}\label{sec:algdesc}

We cannot use Algorithm \ref{alg:sparsify} in the streaming model since it
is not possible to compute the strong connectivity of an edge in $G$
without storing all the data. The overall idea would be to use a strongly
recursive process, where we use an estimation of the connectivity based on
the current sparsification and show that subsequent addition of edges does
not impact the process. The modification is not difficult to state, which
makes us believe that such a modification is likely to find use in
practice. The nontrivial part of the algorithm is in the analysis,
ensuring that the various dependencies being built into the process does
not create a problem. For completeness the modifications are 
presented in Algorithm
\ref{alg:streamsparsify}.

\medskip
\begin{algorithm}[H]\label{alg:streamsparsify}
\noindent{\bf Stream-Sparsification}\\
  \SetLine
  \KwData{The sequence of edges $e_1,e_2,\cdots,e_m$}
  \KwResult{Sparsified graph $H$}
    $H\leftarrow\emptyset$\;
    \ForEach{$e$}{
      compute the connectivity $c_e$ of $e$ in $H$\;
      $p_e=\min\{\rho/c_e,1\}$\;
	   add $e$ to $H$ with probability $p_e$ and weight $1/p_e$\;
    }
  \caption{Streaming Sparsification Algorithm}
\end{algorithm}

We use $\rho=32((4+d)\ln n+\ln m)(1+\epsilon)/\epsilon^2$ given
$\epsilon>0$; once again $d$ is a constant which determines the
probability of success. We prove two theorems for Algorithm
\ref{alg:streamsparsify}. The first theorem is about the approximation
ratio and the second theorem is about its space requirement. For the
simplicity of proof, we only consider sufficiently small $\epsilon$.

\begin{theorem}\label{thm:correctness} Given $\epsilon>0$, $H$ is a
  sparsification, that is $H\in(1\pm\epsilon)G$, with probability
  $1-\O(1/n^d)$.
\end{theorem}

\begin{theorem}\label{thm:space} If $H\in(1\pm\epsilon)G$, $H$ has $\O(n(d\log n+\log m)(\log m-\log n)(1+\epsilon)^2/\epsilon^2)$ edges.
\end{theorem}

We use a sequence of ideas similar to that in Benczur and
Karger~\cite{benczur1996}. Let us first discuss the proof in
\cite{benczur1996}.

In that paper, Theorem \ref{thm:benczur_errorbound} is proved on three
steps. First, the result of Karger~\cite{karger1994}, on uniform sampling
is used. This presents two problems. The first is that they
need to know the value of minimum cut to get a constant error bound. The
other is that the number of edges sampled is too large. In worst case,
uniform sampling gains only constant factor reduction in number of edges.

To solve this problem, Benczur and Karger~\cite{benczur1996} decompose a
graph into $k$-strong connected components. In a $k$-strong connected
component, minimum-cut is at least $k$ while the maximum number of edges
in $k$-strong connected component(without $(k+1)$-strong connected
component as its subgraph) is at most $kn$. They used the uniform sampling
for each component and different sampling rate for different components.
In this way, they guarantee the error bound for every cut.

We cannot use Karger's result \cite{karger1994} directly to prove our
sparsification algorithm because the probability of sampling an edge
depends on the sampling results of previous edges. We show that the error
bound of a single cut by a suitable bound on the martingale process. Using that
we prove that if we do not make an error until $i^{\rm th}$ edge, we
guarantee the same error bound for every cut after sampling $(i+1)^{\rm
  th}$ edge with high probability. Using union bound, we prove that our
sparsification is good with high probability.

\section{Proof of Theorem \ref{thm:correctness}}

\subsection{Single Cut}

We prove Theorem \ref{thm:correctness} first. First, we prove the
error bound of a single cut in Lemma \ref{thm:errorbound_singlecut}. The
proof will be similar to that of Chernoff bound~\cite{chernoff1952}. $p$
in Lemma \ref{thm:errorbound_singlecomponent} is a parameter and we use
different $p$ for different strong connected components in the later
proof.

\begin{lemma}\label{thm:errorbound_singlecut} Let
  $C=\{e_{i_1},e_{i_2},\cdots,e_{i_l}\}$ with $i_1<i_2<\cdots<i_l$ be a
  cut in a graph $G$ such that $w_G(e_{i_j})\leq 1$ and $VAL(C,G)=c$. The
  index of the edges corresponds to the arrival order of the edges in the
  data stream. Let $A_{C}$ be an event such that $p_e\geq p$ for all $e\in
  C$. Let $H$ be a sparsification of $G$. Then,
  $\Prob[A_{C}\wedge(|VAL(C,H)-c|>\beta c)]< 2\exp(-\beta^2pc/4)$ for any
  $0<\beta\leq 2e-1$.
\end{lemma}

Let $X_j=pw_H(e_{i_j})$ and $\mu_j=\E[X_j]=pw_G(e_{i_j})$. Then, $|VAL(C,H)-c|>\beta c$ if and only if $|\sum_j X_j-pc|>\beta pc$. As already mentioned, we cannot apply Chernoff bound because there are two problems:
\begin{enumerate}
\item $X_j$ are not independent from each other and
\item values of $X_j$ are not bounded.
\end{enumerate}
The second problem is easy to solve because we have $A_{C}$. Let $Y_j$ be random variables defined as follows:
$$
Y_j = \left\{ 
  \begin{array}{ll}
    X_j & {\rm if~}p_{e_{i_j}}\geq p \\
    \mu_j & {\rm otherwise.}
  \end{array}
\right.
$$
If $A_{C}$ happens, $Y_j=X_j$. Thus,
\begin{eqnarray}
\Prob[A_{C}\wedge(|VAL(C,H)-c|>\beta c)]
 & = & \Prob[A_{C}\wedge(|\sum_j X_j-\sum_j \mu_j|>\beta pc)] \nonumber \\
 & = & \Prob[A_{C}\wedge(|\sum_j Y_j-\sum_j \mu_j|>\beta pc)] \nonumber \\
 & \leq & \Prob[|\sum_j Y_j-\sum_j \mu_j|>\beta pc] \label{eqn:conclusion}
\end{eqnarray}

The proof of (\ref{eqn:conclusion}) is similar to Chernoff bound~\cite{chernoff1952}. However, since we do not have independent Bernoulli random variables, we need to prove the upperbound of $\E[\exp(t\sum_j Y_j)]$ given $t$. We start with $\E[\exp(tY_j)]$.

\begin{lemma}\label{thm:yibound} $\E[\exp(tY_j)|\hijp]\leq \exp(\mu_j(e^t-1))$ for any $t$ and $\hijp$.
\end{lemma}

\begin{proof} There are two cases. Given $\hijp$, $p_{e_{i_j}}\geq p$ or $p_{e_{i_j}}<p$. At the end of each case, we use the fact that $1+x<e^x$.

Case 1 : If $p_{e_{i_j}}<p$, $Y_j=\mu_j$.
\begin{eqnarray}
\E[\exp(tY_j)|\hijp]
 & = & \exp(t\mu_j) \nonumber \\
 & < & \exp(\mu_j(e^t-1)). \nonumber
\end{eqnarray}

Case 2 : If $\peij\geq p$, $Y_j=X_j$. So $\E[\exp(tY_j)|\hijp]=\peij\exp(t\mu_j/\peij)+(1-\peij)$. Let $f(x)=x\exp(t\mu_j/x)+(1-x)$. Observe that $f'(x)\leq 0$ for $x>0$. So $f(x)$ is decreasing function. Also we have $\mu_j=pw_G(e_{i_j})\leq p\leq \peij$ since  $w_G(e_{i_j})\leq 1$. Hence,
$$\peij\exp(t\mu_j/\peij)+(1-\peij)\leq \mu_j\exp(t)+(1-\mu_j).$$
Therefore,
\begin{eqnarray}
\E[\exp(tY_j)|\hijp] & \leq & \mu_j(\exp(t)-1)+1 \nonumber \\
                     & \leq & \exp(\mu_j(e^t-1)). \nonumber
\end{eqnarray}
From case 1 and 2, $\E[\exp(tY_j)|\hijp]\leq\exp(\mu_j(e^t-1))$ for any $\hijp$.
\end{proof}

Now, we prove the upperbound of $\E[\exp(t\sum_j Y_j)]$.

\begin{lemma}\label{thm:sibound} Let $S_j=\sum_{k=j}^{l}Y_k$. For any $t$ and $\hijp$, $\E[\exp(tS_j)|\hijp]\leq\exp(\sum_{k=j}^{l}\mu_j(e^t-1))$.
\end{lemma}

\begin{proof} We prove by induction. For $j=l$, $\E[\exp(tS_j)|\hijp]=\E[\exp(tY_l)|\hijp]\leq\exp(\mu_l(e^t-1))$ by Lemma \ref{thm:yibound}.

Assume that $\E[\exp(tS_{j+1})|\hijs]\leq\exp(\sum_{k=j+1}^{l}\mu_k(e^t-1))$ for any $\hijs$. Then,
\begin{eqnarray}
\E[\exp(tS_j)|\hijp]
 & = & \sum_{y} \Prob[Y_j=y|\hijp] \sum_{\hijs} \E[\exp(t(y+S_{j+1}))|\hijs]\Prob[\hijs|Y_j=y,\hijp] \nonumber \\
 & = & \sum_{y} \exp(ty) \Prob[Y_j=y|\hijp] \sum_{\hijs} \E[\exp(tS_{j+1})|\hijs]\Prob[\hijs|Y_j=y,\hijp] \nonumber \\
 & \leq & \sum_{y} \Prob[Y_j=y|\hijp] \exp\left(\sum_{k=j+1}^{l}\mu_k(e^t-1)\right) \nonumber \\
 & = & \exp\left(\sum_{k=j+1}^{l}\mu_k(e^t-1)\right) \E[Y_j|\hijp] \nonumber \\
 & \leq & \exp\left(\sum_{k=j}^{l}\mu_k(e^t-1)\right) \nonumber
\end{eqnarray}
Therefore, $\E[\exp(tS_j)|\hijp]\leq\exp(\sum_{k=j}^{n}\mu_k(e^t-1))$ for any $\hijp$ and $t$.
\end{proof}

Now we prove Lemma \ref{thm:errorbound_singlecut}. Remember that we only need to prove $\Prob[|\sum_j Y_j-pc|>\beta pc]< 2\exp(-\beta^2pc/4)$ by (\ref{eqn:conclusion}).

\begin{proof}[Proof of Lemma \ref{thm:errorbound_singlecut}] Let $S=S_1=\sum_j Y_j$ and $\mu=\sum_j \mu_j=pc$. We prove in two parts: $\Prob[S>(1+\beta)\mu]\leq \exp(-\beta^2\mu/4)$ and $\Prob[S<(1-\beta)\mu]\leq \exp(-\beta^2\mu/4)$.

We prove $\Prob[S>(1+\beta)\mu]<\exp(-\beta^2\mu/4)$ first. By applying Markov's inequality to $\exp(tS)$ for any $t>0$, we obtain
\begin{eqnarray}
\Prob(S>(1+\beta)\mu)
 & < & \frac{\E[\exp(tS)]}{\exp(t(1+\beta)\mu)} \nonumber \\
 & \leq & \frac{\exp(\mu(e^t-1))}{\exp(t(1+\beta)\mu)}. \nonumber
\end{eqnarray}
The second line is from Lemma \ref{thm:sibound}. From this point, we have
identical proof as Chernoff bound~\cite{chernoff1952} that gives us bound
$\exp(-\beta^2\mu/4)$ for $\beta<2e-1$. To prove that
$\Prob[S<(1-\beta)\mu]<\exp(-\beta^2pc/4)$ we applying Markov's inequality
to $\exp(-tS)$ for any $t>0$, and proceed similar to above.
Using union bound to these two bounds, we obtain a bound of $2\exp(-\beta^4\mu/4)$.
\end{proof}

\subsection{$k$-strong Connected Component}

Now we prove the following lemma given a $k$-strong connected component and parameter $p$. This corresponds to the proof of uniform sampling method in \cite{karger1994}.

\begin{lemma}\label{thm:errorbound_singlecomponent} Let $Q$ be a $k$-strong component such that each edge has weight at most 1. $H_Q$ is its sparsified graph. Let $\beta=\sqrt{4((4+d)\ln n + \ln m)/pk}$ for some constant $d>0$. Suppose that $A_Q$ be an event such that every edge in $Q$ has sampled with probability at least $p$. Then, $\Prob[A_Q\wedge (H_Q\notin (1\pm\epsilon)Q)]=\O(1/n^{2+d}m)$.
\end{lemma}

\begin{proof} Consider a cut $C$ whose value is $\alpha k$ in $Q$. If $A_Q$ holds, every edge in $C$ is also sampled with probability at least $p$. By Lemma \ref{thm:errorbound_singlecut}, $\Prob[A_Q\wedge |VAL(C,H_Q)-\alpha k|>\beta \alpha k]\leq 2\exp(-\beta^2p\alpha k/4)=2(n^{4+d}m)^{-\alpha}$. Let $P(\alpha)=2(n^{4+d}m)^{-\alpha}$.

Let $F(\alpha)$ be the number of cuts with value less or equal to $\alpha k$. By union bound, we have
$$\Prob[A_Q\wedge (H_Q\notin (1\pm\epsilon)Q)]\leq P(1)F(1)+\int_1^{\infty}P(\alpha)\frac{dF}{d\alpha}d\alpha.$$
The number of cuts whose value is at most $\alpha$ times minimum cut is at most $n^{2\alpha}$. Since the value of minimum cut of $Q$ is $k$, $F(\alpha)\leq n^{2\alpha}$. Since $P$ is a monotonically increasing function, this bound is maximized when $F(\alpha)=n^{2\alpha}$. Thus,
\begin{eqnarray}
\Prob[A_Q\wedge (H_Q\notin (1\pm\epsilon)Q)]
 & \leq & F(1)P(1)+\int_1^{\infty}P(\alpha)\frac{dF}{d\alpha}d\alpha \nonumber \\
 & \leq & n^2P(1)+\int_1^{\infty}P(\alpha)(2n^{2\alpha}\ln n)d\alpha \nonumber \\
 & \leq & \frac{2}{n^{2+d}m}+\int_1^{\infty}\frac{\ln n}{n^{\alpha (2+d)}m^{\alpha}}d\alpha \nonumber \\
 & = & \O(1/n^{2+d}m). \nonumber
\end{eqnarray}
\end{proof}

\subsection{Error Bound for $H_i$ and $H$}

\begin{lemma}\label{thm:errorbound_singlestep} The probability of $i$ being the first integer such that $H_i\notin (1\pm\epsilon)G_i$ is $\O(1/n^dm)$.
\end{lemma}

\begin{proof} If $H_j\in(1\pm\beta)G_j$ for all $j<i$, $c_{e_j}\leq (1+\epsilon)c^{(G_j)}_{e_j}\leq (1+\epsilon)c^{(G_i)}_{e_j}$. Remember that $c^{(G)}_e$ denotes the strong connectivity of $e$ in graph $G$.
\begin{eqnarray}
H_i & = & \sum_{j=-\infty}^{\infty} H_{i,j} \nonumber \\
 & = & \sum_{j=-\infty}^{\infty} \left(H_{i,j}+\frac{1}{2}F_{i,j+1}\right) - \sum_{j=-\infty}^{\infty} \frac{1}{2}F_{i,j+1} \nonumber
\end{eqnarray}
$H_{i,j}+(1/2)F_{i,j+1}$ is a sparsification of $G_{i,j}+(1/2)F_{i,j+1}=F_{i,j}$. $F_{i,j}$ consists of $2^{j-1}$-strong connected components. For every $e\in G_{i,j}$, $c^{(G_i)}_{e}<2^j$. So it is sampled with probability at least $p=\rho/(1+\epsilon)2^j$. If we consider one $2^{j-1}$-strong connected component and set $\rho=32((4+d)\ln n+\ln m)(1+\epsilon)/\epsilon^2$, by Lemma \ref{thm:errorbound_singlecomponent}, every cut has error bound $\epsilon/2$ with probability at least $1-\O(1/n^{2+d}m)$. Since there are less than $n^2$ such distinct strong connected components, with probability at least $1-\O(1/n^dm)$, $H_{i,j}+(1/2)F_{i,j+1}\in(1\pm\beta)F_{i,j}$ for every $i,j$. Hence,
\begin{eqnarray}
H_i & \in & \sum_{j=-\infty}^{\infty}(1\pm\epsilon/2) F_{i,j} - \sum_{j=-\infty}^{\infty} \frac{1}{2}F_{i,j+1} \nonumber \\
  & \subseteq & (2\pm\epsilon)G_i - G_i \nonumber \\
  & = & (1\pm\epsilon)G_i. \nonumber
\end{eqnarray}
Therefore, $\Prob[(\forall j<i.H_j \in (1\pm\epsilon)G_j)\wedge(H_i\notin (1\pm\epsilon)G_i)]=\O(1/n^dm)$.
\end{proof}

From Lemma \ref{thm:errorbound_singlestep}, Theorem \ref{thm:correctness} is obvious. $\Prob[H\notin (1\pm\epsilon)G]\leq \sum_{i=1}^{m}\Prob[(\forall j<i.H_j \in (1\pm\epsilon)G_j)\wedge(H_i\notin (1\pm\epsilon)G_i)]=\O(1/n^d)$.

\section{Proof of Theorem \ref{thm:space}}

For the proof of Theorem \ref{thm:space}, we use the following property of strong connectivity.

\begin{lemma}\label{thm:connectivity_bound}~\cite{benczur1996} If the total edge weight of graph $G$ is $n(k-1)$ or higher, there exists a $k$-strong connected components.
\end{lemma}

\begin{lemma} $H\in(1\pm\epsilon)G$, total edge weight of $H$ is at most $(1+\epsilon)m$.
\end{lemma}

\begin{proof} Let $C_v$ be a cut $(\{v\},V-\{v\})$. Since $H\in(1\pm\epsilon)G$,
$VAL(C_v,H)\leq(1+\epsilon)VAL(C_v,G)$. Total edge weight of $H$ is $(\sum_{v\in V} VAL(C_v,H))/2$ since each edge is counted for two such cuts. Similarly, $G$ has $(\sum_{v\in V} VAL(C_v,H))/2=m$ edges. Therefore, if $H\in(1\pm\epsilon)G$, total edge weight of $H$ is at most $(1+\epsilon)m$.
\end{proof}

Let $E_k=\{e:e\in H{\rm~and~}c_e\leq k\}$. $E_k$ is a set of edges that sampled with $c_e=k$. We want to bound the total weight of edges in $E_k$.

\begin{lemma} $\sum_{e\in E_k} w_H(e)\leq n(k+k/\rho)$.
\end{lemma}

\begin{proof} Let $H'$ be a subgraph of $H$ that consists of edges in $E_k$. $H'$ does not have $(k+k/\rho+1)$-strong connected component. Suppose that it has. Then there exists the first edge $e$ that creates a $(k+k/\rho+1)$-strong connected component in $H'$. In that case, $e_i$ must be in the $(k+k/\rho+1)$-strong connected component. However, since weight $e$ is at most $k/\rho$, that component is at least $(k+1)$-strong connected without $e$. This contradicts that $c_e\leq k$. Therefore, $H'$ does not have any $(k+k/\rho+1)$-strong connected component. By Lemma \ref{thm:connectivity_bound}, $\sum_{e\in E_k} w_H(e)\leq n(k+k/\rho)$.
\end{proof}

Now we prove Theorem \ref{thm:space}.

\begin{proof}[Proof of Theorem \ref{thm:space}] If the total edge weight
  is the same, the number of edges is maximized when we sample edges with
  smallest strong connectivity. So in the worst case,

$$\sum_{e\in E_k-E_{k-1}} w_H(e)=nk(1+\rho)-n(k-1)(1+\rho)=n(1+\rho).$$ 
In that case, $k$ is at most $(1+\epsilon)m/n(1+1/\rho)$. Let this value be $k_m$. Then, total number of edges in $H$ is
\begin{eqnarray}
\sum_{i=1}^{k_m} \frac{n(1+1/\rho)}{i/\rho}
 & = & n(\rho+1) \sum_{i=1}^{k_m} \frac{1}{i} \nonumber \\
 & = & O(n(\rho+1) \log(k_m)) \nonumber \\
 & = & O(n\rho(\log m-\log n)) \nonumber \\
 & = & O(n(d\log n+\log m)(\log m-\log n)(1+\epsilon)^2/\epsilon^2). \nonumber
\end{eqnarray}
\end{proof}

\section{Space Lower bounds}

First, we prove a simple space lowerbound for weighted graphs, where the 
lowerbound depends on $\epsilon$.

\begin{theorem}\label{thm:lowerbound} For $0<\epsilon<1$, $\Omega(n(\log C+\log \frac{1}{\epsilon}))$ bits are required in order to sparsify every cut of a weighted graph within $(1\pm\epsilon)$ factor where $C$ is maximum edge weight and $1$ is minimum edge weight.
\end{theorem}

\newcommand{\factor}{{\left(\frac{1+\epsilon}{1-\epsilon}\right)}}
\begin{proof} Let $F$ be a set of graphs such that there is a center node $u$ and other nodes are connected to $u$ by an edge whose weight is one of $1,\factor,\factor^2,\cdots,C$. Then, $|F|=(\log_\factor C)^{n-1}$. For $G,G'\in F$, they must have different sparsifications. So we need $\Omega(\log |F|)$ bits for sparsfication. It is easy to show that $\log |F|=\Omega(n(\log C+\log \frac{1}{\epsilon}))$.
\end{proof}

Now we use the same proof idea for unweighted simple graphs. Since we cannot assign weight as we want, we use $n/2$ nodes as a center instead of having one center node. In this way, we can assign degree of a node from $1$ to $n/2$.

\begin{theorem}\label{thm:unweighted_lowerbound} For $0<\epsilon<1$, $\Omega(n(\log n+\log \frac{1}{\epsilon}))$ bits are required in order to sparsify every cut of a graph within $(1\pm\epsilon)$.
\end{theorem}

\begin{proof} Consider bipartite graphs where each side has exactly $n/2$ nodes and each node in one side has a degree $1,\factor,\factor^2,\cdots,$ or $n/2$. For each degree assignment, there exists a graph that satisfies it. Let $F$ be a set of graphs that has different degree assignments. Then, $|F|=\left(\log_\factor \frac{n}{2}\right)^{n-1}$. $G,G'\in F$ cannot have the same sparsification. So we need at least $\Omega(\log |F|)=\Omega(n(\log n+\log \frac{1}{\epsilon}))$ bits.
\end{proof}

Another way of viewing the above claim is a direct sum construction, where
we need to use $\Omega(\log \frac1\epsilon)$ bits to count upto a
precision of $(1+\epsilon)$.

\section{Conclusion and Open Problems}

We presented a one pass semi-streaming algorithm for the adversarially
ordered data stream model which uses $O(n(d\log n+\log m)(\log m-\log
n)(1+\epsilon)^2/\epsilon^2)$ edges to provide $\epsilon$ error bound for
cut values with probability $1-O(1/n^d)$. If the graph does not have
parallel edges, the space requirement reduces to $O(dn\log^2 n
(1+\epsilon)^2/\epsilon^2)$. We can solve the minimum cut problem or other
problems related to cuts with this sparsification. For the minimum cut
problem, this provides one-pass
$((1+\epsilon)/(1-\epsilon))$-approximation algorithm.

A natural open question is to determine how the space complexity of the
approximation depends on $\epsilon$. Our conjecture is that the bound of
$n/\epsilon^2$ is tight up to logarithmic factors.

\bibliographystyle{plain}

\begin{thebibliography}{10}
\bibitem{AMS96}
Noga Alon, Yossi Matias, and Mario Szegedy.
\newblock The Space Complexity of Approximating the Frequency Moments. 
\newblock {\em J. Comput. Syst. Sci.}, 58(1):137-147, 1999.

\bibitem{benczur1996}
Andr\'{a}s~A. Bencz\'{u}r and David~R. Karger.
\newblock Approximating s-t minimum cuts in \, {O}(n2) time.
\newblock In {\em STOC '96: Proceedings of the twenty-eighth annual ACM
  symposium on Theory of computing}, pages 47--55, New York, NY, USA, 1996.
  ACM.

\bibitem{chekuri1997}
Chandra~S. Chekuri, Andrew~V. Goldberg, David~R. Karger, Matthew~S. Levine, and
  Cliff Stein.
\newblock Experimental study of minimum cut algorithms.
\newblock In {\em SODA '97: Proceedings of the eighth annual ACM-SIAM symposium
  on Discrete algorithms}, pages 324--333, Philadelphia, PA, USA, 1997. Society
  for Industrial and Applied Mathematics.

\bibitem{chernoff1952}
H.~Chernoff.
\newblock A measure of the asymptotic efficiency for tests of a hypothesis
  based on the sum of observations.
\newblock {\em Annals of Mathematical Statistics}, 23:493--509, 1952.

\bibitem{demetrescu2006}
Camil Demetrescu, Irene Finocchi, and Andrea Ribichini.
\newblock Trading off space for passes in graph streaming problems.
\newblock {\em SODA}, pages 714--723, 2006.

\bibitem{mcgregor2005}
Joan Feigenbaum, Sampath Kannan, Andrew McGregor, Siddharth Suri, and Jian
  Zhang.
\newblock On graph problems in a semi-streaming model.
\newblock {\em Theor. Comput. Sci.}, 348(2):207--216, 2005.

\bibitem{gomory1961}
R.~E. Gomory and T.C. Hu.
\newblock Multi-terminal network flows.
\newblock {\em J. Soc. Indust. Appl. Math.}, 9(4):551--570, 1961.

\bibitem{hao1992}
Jianxiu Hao and James~B. Orlin.
\newblock A faster algorithm for finding the minimum cut in a graph.
\newblock In {\em SODA '92: Proceedings of the third annual ACM-SIAM symposium
  on Discrete algorithms}, pages 165--174, Philadelphia, PA, USA, 1992. Society
  for Industrial and Applied Mathematics.

\bibitem{henzinger1998}
M.~Henzinger, P.~Raghavan, and S.~Rajagopalan.
\newblock Computing on data streams, 1998.

\bibitem{karger1993}
David~R. Karger.
\newblock Global min-cuts in rnc, and other ramifications of a simple min-out
  algorithm.
\newblock In {\em SODA '93: Proceedings of the fourth annual ACM-SIAM Symposium
  on Discrete algorithms}, pages 21--30, Philadelphia, PA, USA, 1993. Society
  for Industrial and Applied Mathematics.

\bibitem{karger1994}
David~R. Karger.
\newblock Random sampling in cut, flow, and network design problems.
\newblock In {\em STOC '94: Proceedings of the twenty-sixth annual ACM
  symposium on Theory of computing}, pages 648--657, New York, NY, USA, 1994.
  ACM.

\bibitem{karger2000}
David~R. Karger.
\newblock Minimum cuts in near-linear time.
\newblock {\em J. ACM}, 47(1):46--76, 2000.

\bibitem{karger1996}
David~R. Karger and Clifford Stein.
\newblock A new approach to the minimum cut problem.
\newblock {\em J. ACM}, 43(4):601--640, 1996.

\bibitem{mcgregor2005b}
Andrew McGregor. 
\newblock Finding Graph Matchings in Data Streams. 
\newblock {\em Proc. of APPROX-RANDOM}, pages 170--181, 2005.

\bibitem{muthu}
S. Muthukrishnan.
\newblock Data streams: Algorithms and Applications.
\newblock{\em Now publishers}, 2006.

\bibitem{MP80}
J. Ian Munro and Mike Paterson.
\newblock Selection and Sorting with Limited Storage. 
\newblock {\em Theor. Comput. Sci.}, 12: 315-323, 1980.

\bibitem{spielman2008}
Daniel~A. Spielman and Nikhil Srivastava.
\newblock Graph sparsification by effective resistances.
\newblock In {\em STOC '08: Proceedings of the 40th annual ACM symposium on
  Theory of computing}, pages 563--568, New York, NY, USA, 2008. ACM.

\end{thebibliography}

\end{document}